\documentclass[a4paper,11pt]{article}

\usepackage{graphicx}
\usepackage{a4wide}
\usepackage[T1]{fontenc}
\usepackage[utf8]{inputenc}
\usepackage{amsmath, amssymb, amsthm}
\usepackage[english]{babel}
\usepackage{eurosym}
\usepackage{tikz}
\usepackage{verbatim,refcount}
\usepackage{fancybox,soul,color}
\usepackage{url}
\usepackage{hyperref}

\usepackage{enumerate}
\usepackage{epsfig}
\usepackage[active]{srcltx}
\usepackage{amsfonts}
\usepackage{amsmath}
\usepackage[mathlines]{lineno}
\usepackage{xspace}
\usepackage{tikz}
\usepackage{multirow,tabu}
\usepackage{slashed}
\usepackage{todonotes}
\usepackage{eucal}
\usepackage{thmtools}
\usepackage{thm-restate}
\usepackage{hyperref}
\usepackage{textpos}
\usepackage{booktabs}
\usepackage{float}
\usepackage{cleveref}
\usepackage{subcaption}

\usepackage[linesnumbered,ruled,vlined]{algorithm2e}

\usetikzlibrary{arrows,shapes,automata,backgrounds,petri,decorations.pathmorphing}

\DeclareMathOperator{\mad}{mad}

\newtheorem{theorem}{Theorem}
\newtheorem{lemma}[theorem]{Lemma}

\def\claimb{\vcenter\bgroup\advance\hsize by -8em\noindent
\refstepcounter{claimb}\ignorespaces\it}
\makeatletter
\def\endclaimb{\rm\egroup\leqno(\theclaim)\global\@ignoretrue}
\makeatother

    {\noindent \emph{Proof.} {}{#1}{}}{\hfill
    $\Diamond$\vspace{1em}}

\newcounter{rulecnt}
\newcounter{confcnt}
\newcommand{\anonymous}[1]{#1}

\newcommand{\reftheorem}[1]{Theorem \ref{#1}}
\newcommand{\reflemma}[1]{Lemma \ref{#1}}

\begin{document}

\title{The planar edge coloring theorem of Vizing in $O(n\log n)$ time}

\anonymous{
\date{\today}

\author{Patryk Jędrzejczak and {\L}ukasz Kowalik\thanks{Institute of Informatics, University of Warsaw, Poland (\texttt{kowalik@mimuw.edu.pl}). This work is a part of project BOBR that has received funding from the European Research Council (ERC) under the European Union’s Horizon 2020 research and innovation programme (grant agreement No. 948057).}}
}

\maketitle

\anonymous{
\begin{textblock}{20}(-1.8, 8.3)
	\includegraphics[width=40px]{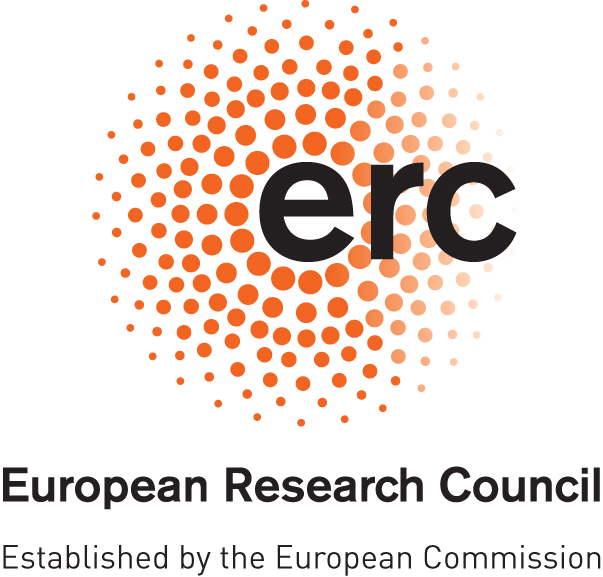}%
\end{textblock}
\begin{textblock}{20}(-2.05, 8.6)
	\includegraphics[width=60px]{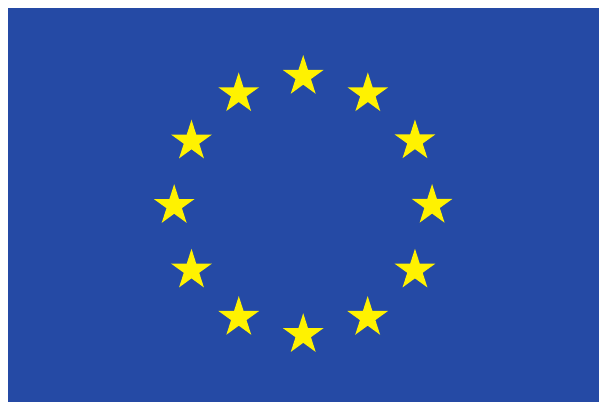}%
\end{textblock}
}

\begin{abstract}
In 1965, Vizing~\cite{vizing:critical} showed that every planar graph of maximum degree $\Delta\ge 8$ can be edge-colored using $\Delta$ colors.
The direct implementation of Vizing's proof gives an algorithm that finds the coloring in $O(n^2)$ time for an $n$-vertex input graph.
Chrobak and Nishizeki~\cite{chrobak-nishizeki} have shown a more careful algorithm, which improves the time to $O(n\log n)$, though only for $\Delta\ge 9$. 
In this paper we extend their ideas to get an algorithm also to the missing case $\Delta=8$. To this end, we modify the original recoloring procedure of Vizing.
This generalizes to bounded genus graphs.
\end{abstract}

\section{Introduction}

Edge-coloring of graph $G$ is an assignment of colors to the edges of $G$ such that incident edges receive different colors. The minimum number of colors that suffice to color graph $G$ is called the \emph{chromatic index} and denoted as $\chi'(G)$. Let $\Delta(G)$ denote the maximum degree among vertices of $G$. Throughout the paper, we write $\Delta$ when there is no ambiguity. For any graph $G$, it is clear that $\chi'(G) \geq \Delta(G)$. Also, Vizing's Theorem~\cite{vizing:first} states that $\chi'(G) \leq \Delta(G) + 1$. If $\chi'(G) = \Delta(G)$, then we say that $G$ is in Class 1, otherwise, that is if $\chi'(G) = \Delta(G)$ + 1, we say that $G$ is in Class 2. 

Determining whether $G$ is Class 1 is an NP-complete problem for general graphs, as shown by Holyer~\cite{holyer}. However, there are interesting variants of the edge-coloring problem that allow for efficient algorithms. Zhou, Nakano and Nishizeki~\cite{ZhouTreewidth} gave an algorithm that uses $\chi'(G)$ colors and runs in linear time for graphs with bounded treewidth. Several works focused on implementing Vizing's Theorem, i.e., $(\Delta + 1)$-edge-coloring~\cite{gabow,sinnamon,BhattacharyaCCS24,bhattacharya2024fasterdelta1edge}, culminating in the recent breakthrough of Assadi, Behnezhad, Bhattacharya, Costa, Solomon and Zhang~\cite{assadi2024vizingstheoremnearlineartime} with a randomized algorithm in time $O(|E(G)|\log\Delta(G))$.
Finally, there are natural graph classes that fit entirely into Class 1 and admit efficient algorithms for edge-coloring with $\Delta$ colors. 
Examples include the $O(m\log\Delta)$-time algorithm for bipartite graphs by Cole, Ost and Shirra~\cite{ColeOstSchirra} and the recent $O(n\log n)$-time algorithm of Kowalik~\cite{DBLP:conf/esa/Kowalik24} for graphs of bounded maximum average degree $\mad(G)$ where $\Delta\ge 2\mad(G)$.

\subsection{Planar graphs}

In this paper, we focus on edge-coloring of planar graphs. Vizing~\cite{vizing:critical} showed that all planar graphs with $\Delta \geq 8$ are in Class 1. He also gave examples of planar graphs with $\Delta \in \{2,3,4,5\}$ belonging to Class 2. Later, Sanders and Zhao~\cite{SandersZhaoDelta7} proved that planar graphs with $\Delta = 7$ are in Class 1. The only remaining case, $\Delta = 6$, remains open. One can turn the proof of Vizing into an $O(n^2)$ algorithm that finds a $\Delta$-edge-coloring of any planar graph with $\Delta \geq 8$. The same can be done with the proof of Sanders and Zhao, giving an $O(n^2)$ algorithm for the case $\Delta = 7$. Currently, these are the best known algorithms for $\Delta$-edge-coloring when $\Delta \in \{7, 8\}$, while more efficient algorithms are known when $\Delta \geq 9$. Chrobak and Nishizeki~\cite{chrobak-nishizeki} provided a $O(n \log n)$ algorithm for the case $\Delta \geq 9$. Later, Cole and Kowalik~\cite{cole-kowalik} used a different approach that gave an $O(n)$ algorithm. 

\subsection{Our result}
In this work, we present an $O(n \log n)$ algorithm for the missing case $\Delta = 8$ of Vizing's original result. Similarly as in the case of the work of Chrobak and Nishizeki~\cite{chrobak-nishizeki}, our algorithm applies to bounded genus graphs, in the sense that in time $O(n\log n)$ the algorithm colors the input graph $G$ using the optimal number of colors $\chi'(G)$.

\subsection{Vizing's proof and the implementation of Chrobak and Nishizeki}

To prove that any planar graph $G$ with $n$ vertices and $\Delta(G) \geq 8$ is in Class 1, Vizing~\cite{vizing:critical} showed that $G$ must contain edge $e$ that can be {\em reduced}, i.e., we can remove $e$ from $G$ and, after coloring $G \setminus \{e\}$ with $\Delta(G)$ colors, we can expand the obtained edge-coloring by $e$ without using more colors. 
In the Vizing's proof, edge $e$ needs to satisfy a condition that involves an upper bound on the number of degree $\Delta$ vertices in the neighborhood of one of its endpoints (see Theorem~\ref{val}). We call such an edge weak in this work.
The existence and reducibility of a weak edge allows for completing the proof of the Vizing's theorem by induction.
This corresponds to an iterative algorithm that begins with an empty edge-coloring and expands it edge by edge.
Coloring a new edge $e$ may require changing the colors of some of the previously colored edges. 
Namely, it may recolor edges incident to $e$ or a single path with alternating colors. The path is called a chain or an $(a,b)$-chain where $a$ and $b$ are the colors that alternate. 
Unfortunately, the chain can be of length $\Omega(n)$, and this is why the direct implementation takes quadratic time. 
Chrobak and Nishizeki found a workaround by reducing a subset of $\Omega(n)$ weak edges simultaneously in time $O(n)$, which leads to the $O(n \log n)$ algorithm.

In more detail, their argument has two ingredients.
First, they proved that if $\Delta(G) \geq 9$, there are $\Omega(n)$ weak edges. 
Second, they show that one can find a set $W$ of $\Omega(n)$ weak edges, which can be colored one by one {\em independently}, i.e., if for an edge $f\in W$ we define $E(f)$ as the set of all recolored edges while coloring an edge $f$ alone, then for any different edges $e_1, e_2 \in W$, there are no $e_1' \in E(e_1)$ and $e_2' \in E(e_2)$ which are the same edge or are incident to each other. 
In particular, this means that recoloring both $e_1$ and $e_2$ results in a proper edge-coloring of $G$, and the same holds if we recolor all the edges from $E_W=\bigcup_{f \in W}E(f)$. 
Moreover, since the sets $E(f)$ are pairwise disjoint, $\sum_{f\in W}|E(f)| = |E_W| \le |E(G)| = O(n)$, i.e., even though a single chain for an edge in $W$ may be of length $\Omega(n)$, all these chains have total length only $O(n)$. It is not hard to color a single edge of $E_W$ in $O(1)$ time, and thus coloring all edges in $W$ takes $O(n)$ time in total.

Let us now explain how the set $W$ is found.
The first step is to find $\Omega(n)$ weak edges that are 2-independent, i.e., not too close to each other: the distance between endpoints of different edges is greater than $2$.
Clearly, this is easy to achieve if $\Delta$ is upper-bounded by a constant, and we can indeed work under this assumption, because when $\Delta$ is large enough, edge-coloring can be found in linear time~\cite{ChrobakY89,cole-kowalik}.
Note that the 2-independency guarantees that the local changes in edge-coloring needed to reduce a weak edge do not interact with each other.
The goal of the next step is to exclude the situation where recoloring a chain of one edge of $W$ interacts with recoloring a chain of another edge. To this end, Chrobak and Nishizeki assign a so-called type to each chosen edge. Edge $e$ is of type $ab$ for different colors $a, b \in [\Delta]$ if the only chain that must be recolored to reduce $e$ is an $(a, b)$-chain. Then, we choose the largest set of edges of the same type, whose size is $\Omega(n)$ since the number of types is bounded. Finally, we filter out at most half of the chosen edges to ensure that each remaining edge requires a different chain to be recolored. $W$ consists of the remaining edges.

\subsection{Our contribution}

The algorithm of Chrobak and Nishizeki does not extend to the case $\Delta(G) = 8$ for one reason: then it is possible that graph $G$ has only $O(1)$ weak edges.
We introduce a new type of edge that can be reduced --- so-called butterfly-like edge.
In Section~\ref{discharging}, we prove using the discharging method that there are either $\Omega(n)$ weak edges or $\Omega(n)$ butterfly-like edges. In the first case, we can use the algorithm provided by Chrobak and Nishizeki. In the second case, our approach is similar. We also base our solution on reducing a fraction of all butterfly-like edges at once in time $O(n)$. The challenge is to show that such a fraction of edges always exists, how to find it, and how to reduce the edges efficiently.

One difficulty is that a butterfly-like edge is reduced by means of one of seven different cases, depending on the coloring of neighboring edges. Thus, controlling interactions between recolorings required to color a set of butterfly-like edges of the same type would require considering ${7 \choose 2}$ situations. We overcome this by incorporating the case into the type. Thus, the number of types is increased, but still bounded, and when we prove that edges in the chosen subset of butterfly-like edges (of the same type) can be colored independently, it is enough to focus on one case of the reducibility procedure at a time only, so $7$ instead of ${7 \choose 2}$ situations. The details are presented in Section~\ref{coloring-algorithm}.

\subsection{Terminology}

We use standard terminology and notation concerning graph theory throughout the paper. Here, we recall the notions used less often and define those specific to this work.

Let $v$ be a vertex of a graph $G$. The degree of $v$ is denoted by $d_G(v)$ or $d(v)$ when $G$ is clear. We call $v$ a $d$-vertex when $d(v) = d$. Similarly, we call a neighbor $w$ of $v$ a $d$-neighbor when $d(w) = d$. By $N_G(v)$, we mean the set of neighbors of $v$, and by $N_G(X)$ for some $X \subseteq V(G)$, we mean $\bigcup_{v \in X} N_G(v)$.

We define the distance between two vertices $v_1, v_2$ of a graph as the number of edges of the shortest path from $v_1$ to $v_2$. If $v_1$ and $v_2$ are in different components, the distance is $\infty$. The distance between vertex $v$ and edge $e$ is the smallest possible distance between $v$ and an endpoint of $e$. The distance between two edges $e_1, e_2$ is the smallest possible distance between an endpoint of $e_1$ and an endpoint of $e_2$. We also call two edges \emph{$k$-independent} if the distance between them is greater than $k$. Similarly, two sets of edges 
$E_1, E_2$ are $k$-independent if for every $e_1 \in E_1$ and $e_2 \in E_2$ edges $e_1, e_2$ are $k$-independent.

We represent edge-coloring of graph $G$ as a function $\pi : E(G) \to \mathbb{N}$ which assigns different values (called \emph{colors}) to incident edges. By \emph{$D$-edge-coloring}, we mean an edge-coloring that uses only colors from $[D] = \{k \in \mathbb{N} : 1 \leq k \leq D\}$. We also use the notion of \emph{partial $D$-edge-coloring}, which is a function $\pi : E(G) \to [D] \cup \{\bot\}$ where $\bot$ represents an uncolored edge and $\pi |_{\pi^{-1}([D])}$ is an edge-coloring.

For a graph $G$, its partial $D$-edge-coloring $\pi$, and $v \in V(G)$, we define the set of \emph{used} colors at $v$ as $\pi(v) = \{\pi(x) : x \in N_G(v)\} \setminus \{\bot\}$ and the set of \emph{free} colors at $v$ as $\bar{\pi}(v) = [D] \setminus \pi(v)$.

Given a partial $D$-edge-coloring $\pi$ of graph $G$ and two different colors $a, b \in [D]$, let $G_{ab}$ be the subgraph of $G$ induced by the edges colored either $a$ or $b$ with respect to $\pi$. $G_{ab}$ is a disjoint union of paths and cycles since its maximum degree is at most two. We call a component of $G_{ab}$ an \emph{$(a, b)$-chain}. If we know that an $(a, b)$-chain is a cycle, we sometimes call it an $(a, b)$-cycle. We call any subgraph of $G$ a \emph{chain} if it is an $(a, b)$-chain for some different colors $a, b$. We note that for any $(a, b)$-chain if we recolor all edges colored $a$ to $b$ and vice versa, we get a new partial $D$-edge-coloring of $G$. We call this operation \emph{kemping} and say that we \emph{kempe} the involved chain.

\subsection{Organization of the paper}

This paper is organized as follows. In Section~\ref{discharging}, we prove that every planar graph with $\Delta(G) \leq 8$ has $\Omega(n)$ weak edges or $\Omega(n)$ butterfly-like edges. Next, in Section~\ref{coloring-algorithm}, we describe an $O(n \log n)$ algorithm that finds an 8-edge-coloring of any planar graph with $\Delta(G) \leq 8$. Finally, in Section~\ref{further-research}, we suggest some related open problems.

\section{Linear lower bound for the number of reducible edges}
\label{discharging}

First, we formally define weak and butterfly-like edges. Let us recall the classic Vizing Adjacency Lemma.

\begin{theorem}
    \textnormal{(Vizing Adjacency Lemma, VAL~\cite{vizing:critical}).} Let $G$ be a simple graph and let $e = xy$ be an edge such that $x$ has at most $D - d(y) + [d(y) = D]$ neighbors of degree $D$. Then any partial $D$-edge-coloring of $G$ which colors a subset $E_c$ of edges of $G$, $e \not\in E_c$ can be expanded to a partial $D$-edge-coloring that colors $E_c \cup \{e\}$.
\label{val}
\end{theorem}

In this work, we consider only the case $D = 8$ and graphs with $\Delta \leq 8$. An edge $xy$ is called \emph{$x$-weak} when it satisfies the assumption of VAL for $D = 8$, that is, $x$ has at most $8 - d(y) + [d(y) = 8]$ 8-neighbors. If $xy$ is \emph{$x$-weak} or \emph{$y$-weak}, then it is also \emph{weak}. Note that for graphs with $\Delta < 8$ every edge is weak.

Let $G$ be a simple graph containing a subgraph $B \subseteq G$ such that $B$ is isomorphic to $B_1$ or $B_2$, where $B_1$ and $B_2$ are defined as follows. $B_1$ (presented in Figure~\ref{butterfly_subgraph1}) satisfies:
\begin{itemize}
    \item $V(B)$ consists of different vertices $x, y, z, v_1, v_2, v_3$,
    \item $E(B)$ consists of edges $xy, xz, yv_1, yv_3, zv_1, zv_2$, and $xv_i$ for $i \in [3]$,
    \item $d_G(y) = d_G(z) = 3$, $d_G(x) = 8$, and $d_G(v_i) = 8$ for $i \in [3]$.
\end{itemize}

$B_2$ (presented in Figure~\ref{butterfly_subgraph2}) satisfies:
\begin{itemize}
    \item $V(B)$ consists of different vertices $x, y, z, v_1, v_2, v_3, v_4$,
    \item $E(B)$ consists of edges $xy, xz, yv_1, yv_4, zv_2, zv_3$, and $xv_i$ for $i \in [4]$,
    \item $d_G(y) = d_G(z) = 3$, $d_G(x) = 8$, and $d_G(v_i) = 8$ for $i \in [4]$.
\end{itemize}

\def\scalefig{0.7}

\begin{figure}[H]
\centering
\begin{minipage}{.35\textwidth}
\centering
\scalebox{\scalefig}{
\begin{tikzpicture}[scale=1.6]
\tikzstyle{whitenode}=[draw,circle,fill=white,minimum size=14pt,inner sep=0pt]
    \draw (2,-3) node[whitenode] (v3) [label=left:$v_3$] {8}
        -- ++(135:1cm) node[whitenode] (y) [label=left:$y$] {3}
        -- ++(30:1.63cm) node[whitenode] (v1) [label=above:$v_1$] {8}
        -- ++(-90:0.8cm) node[whitenode] (x) [label=below:$x$] {8}
        -- ++(0:1.41cm) node[whitenode] (z) [label=right:$z$] {3}
        -- ++(-135:1cm) node[whitenode] (v2) [label=right:$v_2$] {8};
    \draw (x) edge [] node [label=left:] {} (y);
    \draw (z) edge [] node [label=left:] {} (v1);
    \draw (x) edge [] node [label=left:] {} (v2);
    \draw (x) edge [] node [label=left:] {} (v3);
\end{tikzpicture}}
\caption{Subgraph $B_1$.}
\label{butterfly_subgraph1}
\end{minipage}%
\hspace{1cm}
\begin{minipage}{.35\textwidth}
\centering
\scalebox{\scalefig}{
\begin{tikzpicture}[scale=1.6]
\tikzstyle{whitenode}=[draw,circle,fill=white,minimum size=14pt,inner sep=0pt]
    \draw (2,-3) node[whitenode] (v4) [label=left:$v_4$] {\footnotesize{8}}
        -- ++(135:1cm) node[whitenode] (y) [label=left:$y$] {\footnotesize{3}}
        -- ++(45:1cm) node[whitenode] (v1) [label=above:$v_1$] {\footnotesize{8}}
        -- ++(-45:1cm) node[whitenode] (x) [label=above:$x$] {\footnotesize{8}}
        -- ++(45:1cm) node[whitenode] (v2) [label=above:$v_2$] {\footnotesize{8}}
        -- ++(-45:1cm) node[whitenode] (z) [label=right:$z$] {\footnotesize{3}}
        -- ++(-135:1cm) node[whitenode] (v3) [label=right:$v_3$] {\footnotesize{8}};
    \draw (x) edge [] node [label=left:] {} (y);
    \draw (x) edge [] node [label=left:] {} (z);
    \draw (x) edge [] node [label=left:] {} (v3);
    \draw (x) edge [] node [label=left:] {} (v4);
\end{tikzpicture}}
\caption{Subgraph $B_2$.}
\label{butterfly_subgraph2}
\end{minipage}
\end{figure}

Subgraph $B$ is called a \emph{butterfly} and edge $xy$ is called \emph{butterfly-like}.

An edge is called \emph{reducible} when it is weak or butterfly-like. Later, we will show how to reduce butterfly-like edges. Here, we prove that introducing them is sufficient for a linear bound for the number of reducible edges.

\begin{theorem}
Let $G$ be an $n$-vertex planar graph without isolated vertices and such that $\Delta(G) \leq 8$. Then $G$ contains at least $n/1460$ reducible edges. 

More generally, genus $g$ graphs without isolated vertices and with maximum degree at most 8 have at least $n/1460-(g-1)/3$ reducible edges.
\label{linear-bound-reducible-edges}
\end{theorem}

We note that the constant 1/1460 can be improved, but we decided not to do it for the sake of the proof simplicity, since any linear bound is sufficient for us.

\begin{proof}

Let $G=(V,E)$ and denote $m=|E|$. We can assume that $G$ is connected, because otherwise we apply the proof to each component of $G$. 

We can also assume that $\Delta(G) = 8$. Otherwise, all edges of $G$ are weak.

First, assume $G$ contains less than $\tfrac{n}{10}$ vertices of degree 8. Then, since $\Delta(G) = 8$, less than $\tfrac{9n}{10}$ vertices are of degree 8 or have an 8-neighbor. In other words, if $Z \subseteq V$ is a set of all vertices of a degree smaller than 8 that do not have a neighbor of degree 8, then $|Z| \geq \tfrac{n}{10}$. Let $I_Z$ be a maximal independent set in $G[Z]$. By maximality, $Z \subseteq I_Z \cup N_{G[Z]}(I_Z)$. Then, since $\Delta(G[Z]) \leq 7$, $|Z| \leq |I_Z| + 7|I_Z| = 8|I_Z|$, so $|I_Z| \geq \tfrac{|Z|}{8} \geq \tfrac{n}{80}$. Let $x$ be any vertex of $I_Z$ and let $y \in V$ be any of its neighbors. Since $x$ has no $8$-neighbors and $d(y) \leq 8$, edge $xy$ is $x$-weak. We can find a weak edge incident to any vertex in $I_Z$, and all of these edges are different because $I_Z$ is an independent set. Hence, $G$ contains at least $\tfrac{n}{80}$ reducible edges. From now on, we can assume that $G$ contains at least $\tfrac{n}{10}$ vertices of degree 8.

A vertex of $G$ that is an endpoint of a reducible edge is called \emph{reducible}; other vertices are \emph{non-reducible}. Let $R, N \subseteq V$ be the sets of all reducible and non-reducible vertices of $G$, respectively. The number of reducible edges in $G$ is at least $\tfrac{|R|}{2}$. If at least $\tfrac{1}{73}$ of the vertices of degree 8 in $G$ are reducible, then $|R| \geq \tfrac{n}{730}$ and $G$ contains at least $\tfrac{n}{1460}$ reducible edges. Therefore, we can assume that there are at least $\tfrac{72n}{730}$ non-reducible vertices of degree 8 in $G$.

We use the discharging method (see e.g.~\cite{CRANSTON2017766}). Fix any embedding of $G$ into a surface of genus $g$ and let $F$ be the set of faces of graph $G$. Each vertex $v \in V$ receives an initial \emph{charge} $ch(v) = d(v) - 4$ and each face $f \in F$ receives a charge $ch(f) = |f| - 4$ where $|f|$ is the length of the shortest closed walk induced by all edges incident to $f$. Using Euler’s Formula, we can easily bound the total initial charge on $G$:
\begin{equation}
\label{eq:euler}
    \sum_{v\in V} ch(v) + \sum_{f\in F} ch(f) = \sum_{v \in V}(d(v) - 4) + \sum_{f\in F}(|f| - 4) = 2m - 4n + 2m - 4|F| = -8(1-g).
\end{equation}
Our plan is to show that by moving charges in $G$, we can obtain the final charge $ch'$ such that:
\begin{description}
    \item[(F1)] $ch'(v) \geq -12$ for every $v \in R$,
    \item[(F2)] $ch'(v) \geq 0$ for every $v \in N$ such that $d(v) \leq 7$,
    \item[(F3)] $ch'(v) \geq 1/6$ for every $v \in N$ such that $d(v) = 8$,
    \item[(F4)] $ch'(f) \geq 0$ for every $f \in F$.
\end{description} Since the total charge does not change, we can get the claim of the theorem as follows:
\begin{equation}
\label{eq:boundR}
-8(1-g) = \sum_{v \in R} ch'(v) + \sum_{v \in N} ch'(v) + \sum_{f\in F} ch'(f) \geq \sum_{v \in R} ch'(v) + \sum_{v \in N} ch'(v) \geq -12|R| + \frac{72n}{730} \cdot \frac{1}{6}.   
\end{equation}
We get $|R| \geq \tfrac{n}{730}+\tfrac23(1-g)$, so $G$ contains at least $\frac{n}{1460}+\tfrac13(1-g)$ reducible edges, as required (recall that $g=0$ for planar graphs).

Now, we specify the rules of moving the charge, called the discharging rules. Let \emph{triangle face} be a face incident to exactly 3 distinct vertices. Let $(a, b, c)$-triangle be a triangle face incident to distinct vertices $x, y, z$ such that $d(x) = a, d(y) = b, d(z) = c$. Also, for example, let $(a, b, \geq c)$-triangle be a similar triangle face, but $d(z) \geq c$. Additionally, let $(a)$-triangle be a triangle face incident to $x, y, z$ where $\min \{d(x), d(y), d(z)\} = a$. Similar definitions follow for a $(\leq a)$-triangle and a $(\geq a)$-triangle --- we always consider the minimum degree among incident vertices.

\begin{flushleft}
    \textbf{Rule 1.} Each vertex sends 1 to each incident 2-vertex.\\
    \textbf{Rule 2.} Each vertex sends 1/3 to each incident 3-vertex.\\
    \textbf{Rule 3.} Each reducible vertex sends 1 to each incident triangle face.\\
    \textbf{Rule 4.} Each non-reducible $d$-vertex, where $d \in \{5, 6\}$, sends 1/15 to each incident ($d$, 8, 8)-triangle.\\
    \textbf{Rule 5.} Each non-reducible 5-vertex sends 1/3 to each incident (5, 5, 8)-triangle, (5, 6, 6)-triangle, and (5, 6, 7)-triangle.\\
    \textbf{Rule 6.} Each non-reducible 5-vertex sends 1/5 to each incident triangle face not mentioned in rules 4, 5.\\
    \textbf{Rule 7.} Each non-reducible 6-vertex sends 1/2  to each incident (4, 6, 8)-triangle.\\
    \textbf{Rule 8.} Each non-reducible 6-vertex sends 1/3 to each incident triangle face not mentioned in rules 4, 7.\\
    \textbf{Rule 9.} Each non-reducible $d$-vertex, where $d \in \{7, 8\}$, sends 1/2 to each incident $(\leq 4)$-triangle.\\
    \textbf{Rule 10.} Each non-reducible 7-vertex sends 2/5 to each incident (5)-triangle.\\
    \textbf{Rule 11.} Each non-reducible 7-vertex sends 1/3 to each incident triangle face not mentioned in rules 9, 10.\\
    \textbf{Rule 12.} Each non-reducible 8-vertex sends 7/15 to each incident (5)-triangle and (6)-triangle.\\
    \textbf{Rule 13.} Each non-reducible 8-vertex sends 1/3 to each incident triangle face not mentioned in rules 9, 12.\\
\end{flushleft}

Before we prove that the final charge satisfies (F1)-(F4) after following the rules 1-13, we present two helpful lemmas that follow immediately from VAL.

\begin{lemma}
    Let $G$ be a graph such that $\Delta(G) = 8$. For any edge $xy$ of $G$ if $d(x) + d(y) < 10$, then $xy$ is weak.
\label{sum-degrees-weak}
\end{lemma}

\begin{lemma}
    Let $G$ be a graph such that $\Delta(G) = 8$. If $x$ is non-reducible, then for every $y \in N_G(x)$ the number of 8-neighbors of $x$ is at least $9 - d(y) + [d(y) = 8]$.
\label{non-reducible-neighbors-degrees}
\end{lemma}

We start by proving (F4). Let $f \in F$ be any face. Faces never send any charge, so if $|f| \geq 4$, then $ch'(f) \geq ch(f) = |f| - 4 \geq 0$. We must show that if $f$ is a triangle face, then $f$ receives at least 1 unit of charge. If $f$ is incident to a reducible vertex, $f$ receives 1 unit of charge by Rule 3. We can now assume that all vertices incident to the triangle face $f$ are non-reducible, and thus all edges of $f$ are not reducible. By \reflemma{sum-degrees-weak}, we can assume that each edge incident to $f$ has the sum of endpoints' degrees at least 10.

Assume $f$ is a ($\leq 3$)-triangle. Then, $f$ must be a $(\leq 3, \geq 7, \geq 7)$-triangle. By Rule 9, $f$ receives 1 unit of charge.

Assume $f$ is a (4)-triangle. Then, $f$ must be a $(4, \geq 6, \geq 6)$-triangle. When $f$ is a $(4, \geq 7, \geq 7)$-triangle, then by Rule 9, $f$ receives 1 unit of charge. Otherwise, $f$ is a $(4, 6, \geq 6)$-triangle. Note that $f$ cannot be a $(4, 6, \leq 7)$-triangle since, by \reflemma{non-reducible-neighbors-degrees}, at least 5 neighbors of the 6-vertex must be of degree 8. Hence, in the only remaining case $f$ is a (4, 6, 8)-triangle and receives 1 unit of charge by rules 7 and 9.

Assume $f$ is a (5)-triangle. The case analysis below shows that $f$ gets at least 1 unit:
\begin{itemize}
    \item (5, 8, 8)-triangle by rules 4, 12,
    \item (5, 7, 8)-triangle by rules 6, 10, 12,
    \item (5, 7, 7)-triangle by rules 6, 10,
    \item (5, 6, 8)-triangle by rules 6, 8, 12,
    \item (5, 6, 7)-triangle by rules 5, 8, 10,
    \item (5, 6, 6)-triangle by rules 5, 8,
    \item (5, 5, 8)-triangle by rules 5, 12.
\end{itemize}
Note that $f$ cannot be a $(5, 5, \leq 7)$-triangle because, by \reflemma{non-reducible-neighbors-degrees}, any non-reducible 5-vertex with a 5-neigbor has at least four 8-neighbors.

Finally, assume $f$ is a $(\geq 6)$-triangle. If $f$ is a $(6, 8, 8)$-triangle, then $f$ receives $1/15 + 2 \cdot 7/15 = 1$ unit of charge by rules 4 and 12. Otherwise, $f$ receives at least 1/3 of a unit of charge from each incident vertex by rules 8, 11, 12, and 13. In total, $f$ receives at least 1.

From now on, we only need to show that all vertices satisfy conditions (F1)-(F3). We start by proving (F1), i.e., $ch'(v) \geq -12$ for each $v \in R$. Since $v$ is reducible, it sends at most 1 unit of charge to each of its $d(v)$ neighbors by rules 1 and 2. Additionally, by Rule 3, $v$ sends 1 unit of charge to all incident triangle faces, and there are at most $d(v)$ of them. Therefore, $ch'(v) \geq d(v) - 4 - 2d(v) \geq -12$.

Now, we prove (F2), i.e., $ch'(v) \geq 0$ for each $v \in N$ such that $d(v) \leq 7$. Note that by \reflemma{sum-degrees-weak}, all neighbors of $v$ must be of degree at least $10 - d(v)$. In particular, $v$ cannot be of degree one. Also, note that if $d(v) \leq 4$, then $v$ never sends any charge, and if $d(v) \geq 4$, then $v$ never receives any charge.

Assume $d(v) = 2$. Then, $v$ receives 2 units of charge by Rule 1, so $ch'(v) = ch(v) + 2 = d(v) - 4 + 2 = 0$.

Assume $d(v) = 3$. Then, $v$ receives 1 unit of charge in total by Rule 2, so $ch'(v) = ch(v) + 1 = d(v) - 4 + 1 = 0$.

Assume $d(v) = 4$. Observe that $v$ does not send or receive any charge, so $ch'(v) = 0$.

From now on, let $w$ be the neighbor of $v$ with the lowest degree. If there are multiple such neighbors, $w$ is any of them. Recall that $d(w) \geq 10 - d(v)$ by Lemma 1.

Assume $d(v) = 5$ and $d(w) = 5$. By \reflemma{non-reducible-neighbors-degrees}, $v$ has four 8-neighbors. Therefore, the set of triangle faces incident to $v$ consists of at most two incident (5, 5, 8)-triangles and at most three incident (5, 8, 8)-triangles. By rules 5 and 4, respectively, $v$ sends at most $2 \cdot 1/3 + 3 \cdot 1/15 \leq 1$, so $ch'(v) \geq 0$.

Assume $d(v) = 5$ and $d(w) = 6$. By \reflemma{non-reducible-neighbors-degrees}, $v$ has at least three 8-neighbors. Assume first that $v$ is incident to at most four triangle faces. At most one of them receives a charge by Rule 5, and the other by rules 4 and 6, so $v$ sends at most $1/3 + 3 \cdot 1/5 = 14/15 < 1$. Hence, we can assume all the faces incident with $v$ are triangle faces. Let $w, x_1, x_2, x_3, x_4$ be the neighbors of $v$ in the clockwise order in the considered embedding into a genus $g$ surface. Assume first that $d(x_1) < 8$ or $d(x_4) < 8$. By symmetry, we consider only the former case. Then, $v$ sends at most 1/3 to the $(5, 6, d(x_1))$-triangle $vwx_1$ by rules 5 and 6, exactly 1/5 to the $(5, d(x_1), 8)$-triangle $vx_1x_2$ by Rule 6, $2 \cdot 1/15$ to the (5, 8, 8)-triangles $vx_2x_3$ and $vx_3x_4$ by Rule 4, and finally 1/5 to the (5, 6, 8)-triangle $vwx_4$ by Rule 6. Hence, altogether $v$ sends at most $1/3 + 2 \cdot 1/5 + 2/15 = 13/15 < 1$. We are left with the case $d(x_1) = d(x_4) = 8$. By symmetry, $d(x_3) = 8$. The triangle faces $vwx_1$ and $vwx_4$ receive 1/5 by Rule 6, the triangle face $vx_3x_4$ receives 1/15 by Rule 4, and finally the triangle faces $vx_1x_2$ and $vx_2x_3$ receive at most 1/5 by Rule 4 or 6. Thus, $v$ sends at most $4/5 + 1/15 < 1$, as required.

Assume $d(v) = 5$ and $d(w) \geq 7$. Then, $v$ sends at most 1/5 to each incident triangle face by rules 4 and 6, so it sends at most 1 unit in total.

Assume $d(v) = 6$ and $d(w) = 4$. By \reflemma{non-reducible-neighbors-degrees}, all the remaining neighbors of $v$ are of degree 8. Hence, $v$ has at most two incident (4, 6, 8)-triangles and at most four incident (6, 8, 8)-triangles. By rules 7 and 4, respectively, $v$ sends at most $2 \cdot 1/2 + 4 \cdot 1/15 < 2$ units.

Assume $d(v) = 6$ and $d(w) \geq 5$. Then, $v$ sends at most 1/3 of a unit of charge to each incident triangle face by rules 4 and 8. So, it sends at most 2 units in total.

Assume $d(v) = 7$ and $d(w) = 3$. By \reflemma{non-reducible-neighbors-degrees}, all the remaining neighbors of $v$ are of degree 8. Hence, $v$ has at most 2 incident (3, 7, 8)-triangles and at most 5 incident (7, 8, 8)-triangles. In total, $v$ sends at most $1/3 + 2 \cdot 1/2 + 5 \cdot 1/3 = 3$ units of charge by rules 2, 9, and 11, so $ch'(v) \geq 0$.

Assume $d(v) = 7$ and $d(w) = 4$. By \reflemma{non-reducible-neighbors-degrees}, $v$ has at least five 8-neighbors. Let $k$ be the number of $(\leq 5)$-triangles incident to $v$. We can observe that $k \leq 4$ because $v$ has at most two $(\leq 5)$-neighbors. The rest of the triangles incident to $v$ must be $(\geq 6, \geq 6, 7)$-triangles, and there are at most $7 - k$ of them. Therefore, $v$ sends at most $k \cdot 1/2 + (7 - k) \cdot 1/3 = k/6 + 7/3$ of a unit of charge by rules 9, 10, and 11. Since $k \leq 4$, $v$ sends at most 3 units of charge.

Assume $d(v) = 7$ and $d(w) \geq 5$. Then, $v$ sends at most 2/5 of a unit of charge to each incident triangle face by rules 10 and 11, so $v$ sends $7 \cdot 2/5 \leq 3$ units of charge in total. This completes the proof of (F2).

Finally, we prove (F3), i.e., $ch'(v) \geq 1/6$ for each $v \in N$ such that $d(v) = 8$.

Assume $d(v) = 8$ and $d(w) = 2$. By \reflemma{non-reducible-neighbors-degrees}, all the remaining neighbors of $v$ are of degree 8. Note that $w$ can belong to only one triangle face because otherwise $G$ is not simple. Hence, $v$ has at most one incident (2, 8, 8)-triangle and at most six incident (8, 8, 8)-triangles. In total, $v$ sends at most $1 + 1/2 + 6 \cdot 1/3 < 4 - 1/6$ by rules 1, 9, and 13.

Assume $d(v) = 8$ and $d(w) = 3$. By \reflemma{non-reducible-neighbors-degrees}, $v$ has at least six 8-neighbors. Let $t$ be the number of 3-neighbors of $v$. Clearly, $t \in \{1, 2\}$. Let $k$ be the number of $(\leq 6)$-triangles incident to $v$. Since $v$ has at most two $(\leq 6)$-neighbors, we can observe that $k \leq 4$. The rest of the triangle faces incident to $v$ must be $(\geq 7, \geq 7, 8)$-triangles, and there are at most $8 - k$ of them. In total, $v$ sends at most $t \cdot 1/3 + k \cdot 1/2 + (8 - k) \cdot 1/3 = t/3 + k/6 + 8/3$ of a unit of charge by rules 2, 9, 12, and 13. If $t = 1$ or $k \leq 3$, then $v$ sends at most $23/6$ of a unit of charge, so $ch'(v) \geq 1/6$. Otherwise, i.e., if $t = 2$ and $k = 4$, $G$ contains a butterfly as a subgraph, and one of the edges incident to $v$ is butterfly-like, which contradicts $v$ being non-reducible.

Assume $d(v) = 8$ and $d(w) = 4$. By \reflemma{non-reducible-neighbors-degrees}, $v$ has at least five 8-neighbors. Let $k$ be the number of $(\leq 6)$-triangles incident to $v$. Since $v$ has at most three $(\leq 6)$-neighbors, we can observe that $k \leq 6$. The rest of the triangle faces incident to $v$ must be $(\geq 7, \geq 7, 8)$, and there are at most $8 - k$ of them. Hence, $v$ sends at most $k \cdot 1/2 + (8 - k) \cdot 1/3 = k/6 + 8/3$ by rules 9, 12, and 13. Since $k \leq 6$, $v$ sends at most $11/3 \leq 4 - 1/6$ of a unit of charge.

Assume $d(v) = 8$ and $d(w) \geq 5$. Then, $v$ sends at most 7/15 to each incident triangle face by rules 12 and 13. In total, $v$ sends at most $8 \cdot 7/15 \leq 4 - 1/6$. This completes the proof of (F3) and the whole theorem.
\end{proof}

\section{Our coloring algorithm}
\label{coloring-algorithm}

In our solution, we reduce many reducible edges simultaneously. Roughly, it means that we remove some reducible edges from a graph at once, find the edge-coloring of the rest of the graph recursively, add the removed edges back to the graph, and finally expand the edge-coloring to all of them efficiently. 

Let $G = (V, E)$ be an $n$-vertex planar graph such that $\Delta(G) \leq 8$. Let $W \subseteq E$ be a set of uncolored reducible edges that are pairwise 4-independent. Assume all edges of $E \setminus W$ are colored. Below, we show how to choose set $W' \subseteq W$ of size $\Omega(|W|)$ and expand the 8-edge-coloring of $G$ by all edges of $W'$ in time $O(n)$. This procedure is sufficient for our needs, as we can successively apply it to color all edges of $W$ in time $O(n \log |W|)$ --- in each step, the number of uncolored edges of $G$ decreases by a constant factor.

The set $W'$ of edges to reduce will contain either only weak edges or only butterfly-like edges. Chrobak and Nishizeki~\cite{chrobak-nishizeki} have already shown how to reduce weak edges, so we omit details related to weak edges in this work. We will rely on the following theorem that directly follows from the correctness of the algorithm provided by Chrobak and Nishizeki.

\begin{theorem}
    \textnormal{(Chrobak, Nishizeki ~\cite{chrobak-nishizeki}).} Let $G = (V, E)$ be a graph such that $\Delta(G) \leq 8$ and let $W \subseteq E$ be a set of weak edges that are pairwise 2-independent. Also, let $\pi$ be a partial 8-edge-coloring of $G$ such that $\pi^{-1}([8]) = E \setminus W$. We can find in time $O(|E|)$ a partial 8-edge-coloring $\sigma$ of $G$ such that $\sigma^{-1}([8]) = (E \setminus W) \cup W'$ for some $W' \subseteq W$ of size $\Omega(|W|)$.
\label{color-many-weak}
\end{theorem}

In what follows, we will show a similar result for butterfly-like edges.

\subsection{Types}

We start by introducing the notion of a \emph{type} of a butterfly-like edge. Every butterfly-like edge will be of at least one type in our setup. The final algorithm, given an 8-edge-coloring and a set $W$ of butterfly-like edges, will extend the set of colored edges by a subset $W' \subseteq W$ where all the edges in $W'$ are of the same type.

Let $G = (V, E)$ be a graph. Let $\pi$ be a partial 8-edge-coloring of $G$ and let $e =xy \in E$ be an uncolored butterfly-like edge. We define the types below. As we will see, whether $e$ has a specific type depends on $G$ and $\pi$.

\tikzset{snake it/.style={decorate, decoration=snake}}

\begin{figure}[t]
	\centering
	\begin{subfigure}{0.3\textwidth}
		\centering
		\scalebox{\scalefig}{
			\begin{tikzpicture}[scale=1.8]
				\tikzstyle{whitenode}=[draw,circle,fill=white,minimum size=15pt,inner sep=0pt]
				\draw (2,-3) node[whitenode] (y) [label=below:{$b$}] {$y$}
				-- ++(0:1.25cm) node[whitenode] (x) [label=below:$a$] {$x$};
				\draw (x) edge [draw=none] node [below] {$\bot$} (y);
				\node[whitenode] (c1) [above = 1cm of y] {};
				\node[whitenode] (c2) [above = 1cm of x] {};
				\draw (y) edge [snake it] node [left, xshift=-1pt] {$(a, b)$} (c1);
				\draw (x) edge [snake it] node [right] {$(a, b)$} (c2);
		\end{tikzpicture}}
		\caption{Type $1_{ab}$.}
		\label{type1}
	\end{subfigure}
	\begin{subfigure}{0.3\textwidth}
		\centering
		\scalebox{\scalefig}{
			\begin{tikzpicture}[scale=1.8]
				\tikzstyle{whitenode}=[draw,circle,fill=white,minimum size=15pt,inner sep=0pt]
				\draw (2,-3) node[whitenode] (y) [label=below:{$a, b$}] {$y$}
				-- ++(45:1cm) node[whitenode] (z) [] {$z$}
				-- ++(-45:1cm) node[whitenode] (x) [label=below:$c$] {$x$};
				\draw (x) edge [] node [below] {$\bot$} (y);
				\draw (y) edge [draw=none] node [above, xshift=-4pt, yshift=-3pt] {$c$} (z);
				\draw (z) edge [draw=none] node [above, xshift=4pt, yshift=-3pt] {$a$} (x);
				\node[whitenode] (c1) [right = 1cm of z] {};
				\node[whitenode] (c2) [right = 1cm of x] {};
				\draw (z) edge [snake it] node [above, xshift=-1pt, yshift=1pt] {$(a, b)$} (c1);
				\draw (x) edge [snake it] node [above, xshift=-1pt, yshift=1pt] {$(a, b)$} (c2);
		\end{tikzpicture}}
		\caption{Type $2_{ab}$.}
		\label{type2}
	\end{subfigure}
	\begin{subfigure}{0.3\textwidth}
		\centering
		\scalebox{\scalefig}{
			\begin{tikzpicture}[scale=1.8]
				\tikzstyle{whitenode}=[draw,circle,fill=white,minimum size=15pt,inner sep=0pt]
				\draw (2,-3) node[whitenode] (y) [label=below:{$b, c$}] {$y$}
				-- ++(0:1.25cm) node[whitenode] (x) [label=below:{$a$}] {$x$}
				-- ++(0:1.25cm) node[whitenode] (z) [label=below:{$b$}] {$z$};
				\draw (x) edge [draw=none] node [below] {$\bot$} (y);
				\draw (x) edge [draw=none] node [below] {$c$} (z);
				\draw (y) edge [bend left=60, snake it] node [above, xshift=0pt, yshift=3pt] {$(a, b)$} (x);
		\end{tikzpicture}}
		\caption{Type $3_{ab}$.}
		\label{type3}
	\end{subfigure}
	\begin{subfigure}{0.3\textwidth}
		\centering
		\scalebox{\scalefig}{
			\begin{tikzpicture}[scale=1.8]
				\tikzstyle{whitenode}=[draw,circle,fill=white,minimum size=15pt,inner sep=0pt]
				\draw (2,-3) node[whitenode] (y) [label=below:{$a, b$}] {$y$}
				-- ++(0:1.25cm) node[whitenode] (x) [label=below:$c$] {$x$}
				-- ++(45:1cm) node[whitenode] (v) {$v$}
				-- ++(-45:1cm) node[whitenode] (z) [label=below:{$b$}] {$z$};
				\draw (x) edge [draw=none] node [below] {$\bot$} (y);
				\draw (x) edge [draw=none] node [above, xshift=-4pt, yshift=-3pt] {$a$} (v);
				\draw (v) edge [draw=none] node [above, xshift=4pt, yshift=-3pt] {$c$} (z);
				\draw (v) edge [bend right=60, snake it] node [above, xshift=-16pt, yshift=-1pt] {$(a, b)$} (x);
				\draw (x) edge node [below] {} (z);
		\end{tikzpicture}}
		\caption{Type $4_{ab}$.}
		\label{type4}
	\end{subfigure}
	\begin{subfigure}{0.3\textwidth}
		\centering
		\scalebox{\scalefig}{
			\begin{tikzpicture}[scale=1.8]
				\tikzstyle{whitenode}=[draw,circle,fill=white,minimum size=15pt,inner sep=0pt]
				\draw (2,-3) node[whitenode] (y) [label=below:{$a, c$}] {$y$}
				-- ++(0:1.25cm) node[whitenode] (x) [label=below:$b$] {$x$}
				-- ++(45:1cm) node[whitenode] (v) {$v$}
				-- ++(-45:1cm) node[whitenode] (z) [label=below:{$b, c$}] {$z$};
				\draw (x) edge [draw=none] node [below] {$\bot$} (y);
				\draw (x) edge [draw=none] node [above, xshift=-4pt, yshift=-3pt] {$c$} (v);
				\draw (v) edge [draw=none] node [above, xshift=4pt, yshift=-3pt] {$a$} (z);
				\draw (y) edge [bend left=60, snake it] node [above, xshift=0pt, yshift=3pt] {$(a, b)$} (x);
				\draw (x) edge node [below] {} (z);
		\end{tikzpicture}}
		\caption{Type $5_{ab}$.}
		\label{type5}
	\end{subfigure}
	\begin{subfigure}{0.3\textwidth}
		\centering
		\scalebox{\scalefig}{
			\begin{tikzpicture}[scale=1.8]
				\tikzstyle{whitenode}=[draw,circle,fill=white,minimum size=15pt,inner sep=0pt]
				\draw (2,-3) node[whitenode] (y) [label=above:{$a, c, d$}] {$y$}
				-- ++(0:1.25cm) node[whitenode] (x) [label=below:$b$] {$x$}
				-- ++(45:1cm) node[whitenode] (v1) {$v_1$}
				-- ++(-45:1cm) node[whitenode] (z) [label=right:{$b, d$}] {$z$}
				-- ++(-135:1cm) node[whitenode] (v2) {$v_2$};
				\draw (x) edge [] node [label=left:] {} (v2);
				\draw (x) edge [draw=none] node [below] {$\bot$} (y);
				\draw (x) edge [draw=none] node [above, xshift=-4pt, yshift=-3pt] {$c$} (v1);
				\draw (x) edge [draw=none] node [below, xshift=-3pt, yshift=1pt] {$a$} (v2);
				\draw (v1) edge [draw=none] node [above, xshift=4pt, yshift=-3pt] {$a$} (z);
				\draw (v2) edge [draw=none] node [below, xshift=2pt, yshift=1pt] {$c$} (z);
				\draw (y) edge [bend right=30, snake it] node [below, xshift=-2pt, yshift=-6pt] {$(a, b)$} (v2);
				\draw (v1) edge [bend right=60, snake it] node [above, xshift=-16pt, yshift=-1pt] {$(c, d)$} (x);
				\draw (x) edge node [below] {} (z);
		\end{tikzpicture}}
		\caption{Type $6_{abcd}$.}
		\label{type6}
	\end{subfigure}
	
	\caption{Types of butterfly-like edges}
\end{figure}

Edge $e$ has type $1_{ab}$ (see Fig.~\ref{type1}) for different $a, b \in [8]$ if $a \in \bar{\pi}(x)$, $b \in \bar{\pi}(y)$, and $x$ is not an endpoint of the $(a, b)$-chain starting at $y$.

Edge $e$ has type $2_{ab}$ (see Fig.~\ref{type2}) for different $a, b \in [8]$ if $a, b \in \bar{\pi}(y)$, and for some $c \in \bar{\pi}(x) \setminus \{b\}$ there is $z \in N_G(x) \cap N_G(y)$ such that $\pi(yz) = c$, $\pi(xz) = a$, and $xz$ does not belong to an $(a, b)$-cycle.

Edge $e$ has type $3_{ab}$ (see Fig.~\ref{type3}) for different $a, b \in [8]$ if $a \in \bar{\pi}(x)$, $b \in \bar{\pi}(y)$, there is an $(a, b)$-chain with endpoints $x$ and $y$, and for some $c \in \bar{\pi}(y) \setminus \{b\}$ there is $z \in N_G(x)$ such that $b \in \bar{\pi}(z)$ and $\pi(xz) = c$.

Edge $e$ has type $4_{ab}$ (see Fig.~\ref{type4}) for different $a, b \in [8]$ if $a, b \in \bar{\pi}(y)$, there is an $(a, b)$-cycle containing $x$, and for some $c \in \bar{\pi}(x)$ there are different $z, v \in V \setminus \{x, y\}$ such that $xv, zv, xz \in E$, $\pi(xv) = a$, $\pi(zv) = c$, and $b \in \bar{\pi}(z)$.

Edge $e$ has type $5_{ab}$ (see Fig.~\ref{type5}) for different $a, b \in [8]$ if $a \in \bar{\pi}(y)$, $b \in \bar{\pi}(x)$, there is an $(a, b)$-chain with endpoints $x$ and $y$, and for some $c \in \bar{\pi}(y) \setminus \{a\}$ there are different $z, v \in V \setminus \{x, y\}$ such that $xv, zv, xz \in E$, $\pi(xv) = c$, $\pi(zv) = a$, and $b, c \in \bar{\pi}(z)$.

Edge $e$ has type $6_{abcd}$ (see Fig.~\ref{type6}) for different $a, b, c, d \in [8]$ if $a, c, d \in \bar{\pi}(y)$, $b \in \bar{\pi}(x)$, there is an $(a, b)$-chain with endpoints $x$ and $y$, there is a $(c, d)$-cycle containing $x$, and there are different $z, v_1, v_2 \in V \setminus \{x, y\}$ such that $xv_1, xv_2, zv_1, zv_2, xz \in E$, $\pi(xv_1) = c$, $\pi(xv_2) = a$, $\pi(zv_1) = a$, $\pi(zv_2) = c$, and $b, d \in \bar{\pi}(z)$.


Edge $e$ has type 0 if we can color $e$ after recoloring only edges at distance at most 1 from $e$. This type is a bit special because it requires different algorithms depending on $G$ and $\pi$. Note that the property that $e$ has type 0 depends only on the colors of edges at distance at most 2 from $e$ because only these edges can get (re)colored or are incident with an edge that gets (re)colored.

\subsection{The chain data structure}

We first prove that any uncolored butterfly-like edge is of one of the defined types. However, to make the final algorithm efficient, we must also show how to find types of butterfly-like edges quickly. The problem is that types may depend on chains. We must be able to find endpoints of an $(a, b)$-chain to verify that, for example, an edge is of type $3_{ab}$. We must also be able to check whether an $(a, b)$-chain is a cycle to verify that, for example, an edge is of type $4_{ab}$. These checks must be done efficiently, since we aim at finding types of $\Omega(n)$ butterfly-like edges in time $O(n)$. Chains can be of length $\Theta(n)$ and a single chain can be needed to determine the type of $\Omega(n)$ edges, so checking the chain-related conditions naively could take time $\Theta(n^2)$, which would be too slow.

To find types of multiple butterfly-like edges more efficiently, we design a simple data structure called the \emph{chain data structure}. For any graph $G$ and a partial $D$-edge-coloring $\pi$ of $G$, the chain data structure $CDS(G, \pi)$ is defined as follows. For each $v \in V(G)$ and for each different $a, b \in [D]$, the $CDS(G, \pi)$ stores the following information about the $(a, b)$-chain $P$ containing $v$:
\begin{itemize}
    \item the boolean value representing whether $P$ is a cycle, and if not,
    \item the endpoints of $P$.
\end{itemize}
Note that the chain $P$ always exists, though it may have no edges.

We will only use chain data structures when $\Delta(G) \leq 8$ and $D = 8$. With such constraints, the chain data structure has size $O(|V(G)|)$, since it uses $O(1)$ words for each vertex. It can also be computed quickly.

\begin{lemma}
\label{compute-cds}
    Fix $D \in \mathbb{N}$. Given an $n$-vertex graph $G$ such that $\Delta(G) \leq D$ and partial $D$-edge-coloring $\pi$ of $G$, $CDS(G, \pi)$ can be computed in time $O(n)$.
\end{lemma}

\begin{proof}
	For each pair of different colors $a, b \in [D]$ we filter out only edges of color $a$, $b$ obtaining a collection of paths and cycles. For each such path/cycle we update the information of $CDS$ in its vertices.
	
	Fix a pair of different colors $a, b \in [D]$. Recall that $G_{ab}$ is the subgraph of $G$ induced by the edges colored either $a$ or $b$ with respect to $\pi$. We can compute $G_{ab}$ in time $O(n)$. Then, we traverse $G_{ab}$ using an algorithm like DFS. For any component of $G_{ab}$, which is an $(a, b)$-chain, we can easily check whether it is a cycle and remember all visited vertices in preorder (the first and the last being the endpoints if the chain is a path). Finally, we can save the information related to this chain for all visited vertices. All these steps take time $O(|V(G_{ab})|) = O(n)$. We can repeat the procedure for all $\binom{|D|}{2} = O(1)$ pairs of colors $a, b$ in time $O(n)$.
\end{proof}

\subsection{Computing types}

We are ready to show that any butterfly-like edge has some type and we can find it quickly.

\begin{lemma}
\label{types-butterfly}
    Let $G = (V, E)$ be a graph such that $\Delta(G) \leq 8$ and let $\pi$ be its partial 8-edge-coloring. If $e \in E$ is an uncolored butterfly-like edge such that all edges at distance at most 1 from $e$ are colored, then $e$ is of at least one of the types $0, 1_{ab}, 2_{ab}, 3_{ab}, 4_{ab}, 5_{ab}, 6_{abcd}$ for some $a, b, c, d \in [8]$. Moreover, given $CDS(G, \pi)$, we can find a type of $e$ in constant time.
\end{lemma}

\begin{proof}
Let $B \subseteq G$ be any butterfly that makes $e$ the butterfly-like edge. To find it, we only have to look at vertices and edges at distance at most 1 from $e$. Their number is bounded by a constant, so this step takes constant time.

First, we consider the simpler case where $B$ is isomorphic to $B_1$. We show that $e$ has type 0 for any $\pi$. Let us recall that to prove that $e$ has type 0, we must show how to modify $\pi$ by (re)coloring only edges at distance at most 1 to obtain a partial 8-edge-coloring $\sigma$ which expands $\pi$ by $e$ that is $\sigma^{-1}([8]) = \pi^{-1}([8]) \cup \{e\}$. From now on, by $\pi(e) := a$, we mean changing $\pi$ by assigning color $a$ to $e$. By swap$_{\pi}(e_1$, $e_2$), where $e_1, e_2 \in E$ are colored edges, we mean modifying $\pi$ by swapping colors of $e_1$ and $e_2$. We also assume that kemping changes $\pi$ throughout the proof.

Assume that the vertices of $B$ are named just like their corresponding vertices of $B_1$, see Figure~\ref{computing-types-b1-app}. Without loss of generality, assume $\pi(xz) = a$, $\pi(xv_1) = b$, $\pi(xv_2) = c$, $\pi(xv_3) = d$, $\bar{\pi}(x) = \{f\}$, where $a$, $b$, $c$, $d$, $f$ are pairwise distinct.

\begin{figure}[ht]
	\begin{center}
	\scalebox{\scalefig}{
		\begin{tikzpicture}[scale=1.6]
			\tikzstyle{whitenode}=[draw,circle,fill=white,minimum size=14pt,inner sep=0pt]
			\draw (2,-3) node[whitenode] (v3) [label=left:$v_3$] {8}
			-- ++(135:1cm) node[whitenode] (y) [label=left:$y$] {3}
			-- ++(30:1.63cm) node[whitenode] (v1) [label=above:$v_1$] {8}
			-- ++(-90:0.8cm) node[whitenode] (x) [label=below:$x$] {8}
			-- ++(0:1.41cm) node[whitenode] (z) [label=right:$z$] {3}
			-- ++(-135:1cm) node[whitenode] (v2) [label=right:$v_2$] {8};
			\draw (x) edge [] node [label=left:] {} (y);
			\draw (z) edge [] node [label=left:] {} (v1);
			\draw (x) edge [] node [label=left:] {} (v2);
			\draw (x) edge [] node [label=left:] {} (v3);
			\draw (x) edge [draw=none] node [below] {$a$} (z);
			\draw (x) edge [draw=none] node [below] {$\bot$} (y);
			\draw (x) edge [draw=none] node [right,xshift=-1pt] {$b$} (v1);
			\draw (x) edge [draw=none] node [below, yshift=-2pt] {$c$} (v2);
			\draw (x) edge [draw=none] node [below, yshift=1pt] {$d$} (v3);
			\draw (x) edge [draw=none] node [left,pos=0.1] {$f$} (v1);
	\end{tikzpicture}}
	\end{center}
	\caption{Proof of~\reflemma{types-butterfly}. Case when $B$ is isomorphic to $B_1$.}
	\label{computing-types-b1-app}
\end{figure}

If $f \in \bar{\pi}(y)$, then $\pi(xy) := f$ gives $\sigma$. Hence, we can assume $f \in \pi(y)$.

Assume $f \in \bar{\pi}(z)$. If $a \in \bar{\pi}(y)$, then $\pi(xz) := f$ and $\pi(xy) := a$ give $\sigma$. Hence, we can assume $a \in \pi(y)$, which leaves us with two cases. In the first case, when $\pi(yv_1) = f$ and $\pi(yv_3) = a$, we can obtain $\sigma$ with swap$_{\pi}(yv_1, zv_1)$ and $\pi(xy) := f$. In the second case, when $\pi(yv_1) = a$ and $\pi(yv_3) = f$, we can obtain $\sigma$ with swap$_{\pi}(yv_1, zv_1)$, $\pi(xz) := f$, and $\pi(xy) := a$. From now on, we can assume $f \in \pi(z)$, which leaves us with three cases:
\begin{enumerate}
	\item $\pi(yv_3) = f, \pi(zv_1) = f$,
	\item $\pi(yv_3) = f, \pi(zv_2) = f$,
	\item $\pi(yv_1) = f, \pi(zv_2) = f$.
\end{enumerate}

For each of these cases, we show how we can obtain $\sigma$ in a few subcases.

In Case 1, if $\pi(zv_2) \neq b$, we do swap$_{\pi}(xv_1, zv_1)$ and $\pi(xy) := b$. If $\pi(zv_2) = b$ and $\pi(yv_1) \neq c$, we do swap$_{\pi}(xv_1, zv_1)$, swap$_{\pi}(xv_2, zv_2)$, and $\pi(xy) := c$. If $\pi(zv_2) = b$ and $\pi(yv_1) = c$, we do swap$_{\pi}(xv_3, yv_3)$, $\pi(xz) := d$, and $\pi(xy) := a$.

In Case 2, if $\pi(zv_1) = c$, we do swap$_{\pi}(xv_1, zv_1)$, swap$_{\pi}(xv_2, zv_2)$, and $\pi(xy) := b$. If $\pi(yv_1) = c$, we do swap$_{\pi}(yv_1, zv_1)$ to reach the previous subcase. If $\pi(zv_1) \neq c$ and $\pi(yv_1) \neq c$, we do swap$_{\pi}(xv_2, zv_2)$ and $\pi(xy) := c$.

In Case 3, if $\pi(yv_3) \neq c$ and $\pi(zv_1) \neq c$, we do swap$_{\pi}(xv_2, zv_2)$ and $\pi(xy) := c$. If $\pi(yv_3) = c$, we do swap$_{\pi}(xv_1, yv_1)$, $\pi(xz) := b$, and $\pi(xy) := a$. Finally, we have $\pi(yv_3) \neq c$ and $\pi(zv_1) = c$. We consider two further subcases. If $\pi(yv_3) \neq b$, we do swap$_{\pi}(xv_1, zv_1)$, swap$_{\pi}(xv_2, zv_2)$, and $\pi(xy) := b$. If $\pi(yv_3) = b$, we do swap$_{\pi}(xv_1, yv_1)$, swap$_{\pi}(xv_3, yv_3)$, $\pi(xz) := d$, and $\pi(xy) := a$.

\begin{figure}[h]
\centering
	\scalebox{\scalefig}{
	\begin{tikzpicture}[scale=1.6]
		\tikzstyle{whitenode}=[draw,circle,fill=white,minimum size=14pt,inner sep=0pt]
		\draw (2,-3) node[whitenode] (v4) [label=left:$v_4$] {8}
		-- ++(135:1cm) node[whitenode] (y) [label=left:$y$] {3}
		-- ++(45:1cm) node[whitenode] (v1) [label=above:$v_1$] {8}
		-- ++(-45:1cm) node[whitenode] (x) [label=above:$x$] {8}
		-- ++(45:1cm) node[whitenode] (v2) [label=above:$v_2$] {8}
		-- ++(-45:1cm) node[whitenode] (z) [label=right:$z$] {3}
		-- ++(-135:1cm) node[whitenode] (v3) [label=right:$v_3$] {8};
		\draw (x) edge [] node [label=left:] {} (y);
		\draw (x) edge [] node [label=left:] {} (z);
		\draw (x) edge [] node [label=left:] {} (v3);
		\draw (x) edge [] node [label=left:] {} (v4);
		\draw (x) edge [draw=none] node [below] {$a$} (z);
		\draw (x) edge [draw=none] node [below] {$\bot$} (y);
		\draw (x) edge [draw=none] node [above, xshift=2] {$b$} (v1);
		\draw (x) edge [draw=none] node [above, xshift=-2] {$c$} (v2);
		\draw (x) edge [draw=none] node [below, xshift=-2pt] {$d$} (v3);
		\draw (y) edge [draw=none] node [above, pos=0.4] {$f$} (v1);
		\draw (x) edge [draw=none] node [below,pos=0,xshift=-5pt] {$f$} (v3);
	\end{tikzpicture}}
	\caption{Proof of~\reflemma{types-butterfly}. Case when $B$ is isomorphic to $B_2$.}
	\label{computing-types-b2}
\end{figure}

Now, consider the case where $B$ is isomorphic to $B_2$. Assume that the vertices of $B$ are named just like their corresponding vertices of $B_2$, see Figure~\ref{computing-types-b2}. The set $\bar{\pi}(x)$ contains only one color, denoted by $f$. If $f \in \bar{\pi}(y)$, then $\pi(xy) := f$ gives $\sigma$, so $e$ has type 0. Hence, we need to consider only the case where $f \in \pi(y)$. By symmetry, we can assume $\pi(yv_1) = f$.

For any $k \in \bar{\pi}(y)$, we can assume there is an $(f, k)$-chain with endpoints $x$ and $y$. Otherwise, $e$ has type $1_{fk}$.

Without loss of generality, assume $\pi(xz) = a$, $\pi(xv_1) = b$, $\pi(xv_2) = c$, and $\pi(xv_3) = d$, where $a$, $b$, $c$, $d$, $f$ are pairwise distinct.

Assume $\pi(yv_4) \neq a$. If $f \in \bar{\pi}(z)$, setting $\pi(xz) := f$ and $\pi(xy) := a$ gives $\sigma$, so $e$ has type 0. Hence, we can assume $f \in \pi(z)$. There is $k \in [8]$ such that $k \in \pi(x)$, $k \in \bar{\pi}(y)$, and $k \in \bar{\pi}(z)$ because $|\pi(x)| = 7$ and $|\pi(y) \cup \pi(z)| \leq 4$. We have already assumed that there is an $(f, k)$-chain with endpoints $x$ and $y$. Hence, $e$ has type $3_{fk}$. To see this, map the vertices/colors from the definition of $3_{ab}$ as follows: $x \to x, y \to y, z \to z, a \to f, b \to k, c \to a$. From now on, we can assume $\pi(yv_4) = a$.

Since $\pi(y) = \{a, f\}$, we conclude that $b \in \bar{\pi}(y)$. Therefore, for any $k \in \bar{\pi}(y)$, we can assume the $(b, k)$-chain containing $x$ is a cycle. Otherwise, $e$ has type $2_{bk}$. To see this, map the vertices/colors from the definition of $2_{ab}$ as follows: $x \to x, y \to y, z \to v_1, a \to b, b \to k, c \to f$.

Assume $f \in \pi(z)$. By symmetry, we can assume $\pi(zv_2) = f$. If $\pi(zv_3) \neq c$, then swap$_\pi(xv_2, zv_2)$, $\pi(xy) := c$ gives $\sigma$, so $e$ has type 0. If $\pi(zv_3) = c$, then $e$ has type $4_{cb}$ by the assumption that the $(b, c)$-chain containing $x$ is a cycle. To see this, map the vertices/colors from the definition of $4_{ab}$ as follows: $x \to x, y \to y, z \to z, v \to v_2, a \to c, b \to b, c \to f$. Hence, $f \in \bar{\pi}(z)$.

Assume $\pi(zv_3) \neq c$. Let $k = \pi(zv_2)$. Since $k \in \bar{\pi}(y)$, $e$ has type $5_{kf}$ by the assumption that there is an $(f, k)$-chain with endpoints $x$ and $y$. To see this, map the vertices/colors from the definition of $5_{ab}$ as follows: $x \to x, y \to y, z \to z, v \to v_2, a \to k, b \to f, c \to c$. Hence, $\pi(zv_3) = c$.

Assume $\pi(zv_2) \neq d$. Let $k = \pi(zv_2)$. Since $c \in \bar{\pi}(y)$, as in the previous case, $e$ has type $5_{kf}$. To see this, map the vertices/colors from the definition of $5_{ab}$ as follows: $x \to x, y \to y, z \to z, v \to v_3, a \to c, b \to f, c \to d$. Hence, $\pi(zv_2) = d$.

Now, we know that $\pi(zv_3) = c$ and $\pi(zv_2) = d$. By the assumptions that there is an $(f, d)$-chain with endpoints $x$ and $y$, and a $(b, c)$-cycle containing $x$, $e$ has type $6_{dfcb}$. To see this, map the vertices/colors from the definition of $6_{abcd}$ as follows: $x \to x, y \to y, z \to z, v_1 \to v_2, v_2 \to v_3, a \to d, b \to f, c \to c, d \to b$. This completes the proof that $e$ has a type.

Consider an algorithm that follows the proof above and finds the type of $e$. The algorithm only checks a bounded number of conditions involving checking color of an edge or existence of a chain. Checking any of these conditions is easy to do in constant time. In particular, checking any chain-related condition is done using $CDS(G, \pi)$.
\end{proof}

Before we show how we can expand 8-edge-coloring by multiple butterfly-like edges at once, we need to find a way to avoid expanding edge-coloring by edges of the types $2_{ab}$ for different $a, b \in [8]$. The issue with these types is that it is not clear whether we can expand edge-coloring by multiple edges of one of these types efficiently. The solution is that if we ever decide to expand edge-coloring by edges of type $2_{ab}$ for some $a, b \in [8]$, we will instead change the 8-edge-coloring so that all these edges receive type $1_{ab}$, and then proceed as if these edges had type $1_{ab}$ in the beginning.

\begin{lemma}
\label{no-type-2}
    Let $G = (V, E)$ be a graph such that $\Delta(G) \leq 8$ and let $\pi$ be its partial 8-edge-coloring. Let $E_{ab} \subseteq E$ be a set of uncolored butterfly-like edges of type $2_{ab}$ for different $a, b\in [8]$ that are pairwise 2-independent. Then, we can find in time $O(|E_{ab}|)$ a partial 8-edge-coloring $\sigma$ of $G$ such that $\sigma^{-1}([8]) = \pi^{-1}([8])$ and all the edges in $E_{ab}$ have type $1_{ab}$.
\end{lemma}

\begin{proof}
	We provide the algorithm. It begins with $X = E_{ab}$ and $Y = \emptyset$. Then, it moves the edges from $X$ to $Y$ one by one without changing the set of colored edges of $G$. The algorithm also maintains the invariant that all the edges in $X$ are of type $2_{ab}$ while all the edges in $Y$ are of type $1_{ab}$.
	
	Assume the invariant holds. Pick $xy \in X$. Since $xy$ is of type $2_{ab}$, there is a color $c \in \bar{\pi}(x) \setminus \{b\}$ and a vertex $z \in N_G(x) \cap N_G(y)$ such that $\pi(yz) = c$, $\pi(xz) = a$, and $xz$ does not belong to an $(a, b)$-cycle. We can find $c$ and $z$ in $O(1)$ time. Then, we swap the colors of $yz$ and $xz$. Note that we get a proper partial 8-edge-coloring, and there is no $(a, b)$-chain from $x$ to $y$, for otherwise before the swap $xz$ was on an $(a, b)$-cycle. Hence, $xy$ is of type $1_{ab}$.
	
	We claim that after the swap all the edges in $X \setminus \{xy\}$ stay of type $2_{ab}$ while all the edges in $Y$ stay of type $1_{ab}$. Indeed, since the edges in $E_{ab}$ are pairwise 2-independent, the swap of colors does not change the colors of edges at distance 1 from any of the edges in $X \setminus \{xy\}$ or $Y$. Moreover, since after the swap $a$ is free at $x$ and $b$ is free at $y$, no new $(a, b)$-cycle is created, so all the edges in $X \setminus \{xy\}$ are of type $2_{ab}$. Also, an edge $e$ of $Y$ could lose type $1_{ab}$ only if $zy$ (which got colored $a$) connected two $(a, b)$-chains each ending in $e$, which is impossible because $a, b$ were free at $y$. It follows that the invariant defined at the beginning of the proof remains satisfied after one step. The invariant implies that after $|E_{ab}|$ steps we get the desired coloring.
\end{proof}

\subsection{Reducing butterfly-like edges}

There is one last constraint on the butterfly-like edges that we need to introduce to enable expanding 8-edge-coloring by multiple butterfly-like edges efficiently. Let $G = (V, E)$ be a graph such that $\Delta(G) \leq 8$ and let $\pi$ be its partial 8-edge-coloring. For different $a, b \in [8]$, edges $e_1, e_2 \in E$ are \emph{$(a, b)$-chain dependent} when there are different vertices $v_1, v_2 \in V$ at distance at most 1 from endpoints of $e_1$ and $e_2$, respectively, such that there is an $(a, b)$-chain with endpoints $v_1$ and $v_2$ with respect to $\pi$. Otherwise, $e_1$ and $e_2$ are \emph{$(a, b)$-chain independent}. Moreover, if $e_1, e_2 \in E$ are butterfly-like edges of the same type $T$, then they are \emph{$T$-chain independent} when one of the following holds:
\begin{itemize}
    \item $T = \textrm{0}$,
    \item $T \in \{1_{ab}, 3_{ab}, 4_{ab}, 5_{ab}\}$ for some $a, b \in [8]$, and $e_1, e_2$ are $(a, b)$-chain independent,
    \item $T = 6_{abcd}$ for some $a, b, c, d \in [8]$, and $e_1, e_2$ are $(a, b)$-chain independent and $(c, d)$-chain independent.
\end{itemize}

We are ready to show how we can expand 8-edge-coloring by multiple butterfly-like edges of the same type at once.

\begin{lemma}
    Let $G = (V, E)$ be an $n$-vertex graph such that $\Delta(G) \leq 8$ and let $\pi$ be its partial 8-edge-coloring. Let $E_T \subseteq E(G)$ be a set of uncolored butterfly-like edges of the same type $T \in \{0, 1_{ab}, 3_{ab}, 4_{ab}, 5_{ab}, 6_{abcd}\}$ for some $a, b, c, d \in [8]$ that are pairwise 4-independent and $T$-chain independent. Then, in time $O(n)$ we can find a partial 8-edge-coloring $\sigma$ of $G$ such that $\sigma^{-1}([8]) = \pi^{-1}([8]) \cup E_T$.
\label{single-type-color-many}
\end{lemma}

\begin{proof}
For any $e \in E_T$, we will show how to find a partial 8-edge-coloring $\sigma_e$ such that $\sigma_e^{-1}([8]) = \pi^{-1}([8]) \cup \{e\}$. The algorithm will only compute $\sigma_e|_{E_e}$, where $E_e = \{e' \in E : \pi(e') \neq \sigma_e(e')\}$, which will allow us to achieve time $O(|E_e|)$ for every $e \in E_T$.

In the proof, we will also argue that, for every different $e, e' \in E_T$, $\sigma_e$ and $\sigma_{e'}$ are \emph{coherent} that is $E_e \cap E_{e'} = \emptyset$, and $\sigma_e(e_1) \neq \sigma_{e'}(e_1')$ for any pair of incident edges $e_1 \in E_e, e_{1}' \in E_{e'}$.

As a consequence of the edges in $E_T$ being pairwise coherent, we get that partial 8-edge-coloring $\sigma$ defined as

\[
\sigma(e) = \left\{\begin{array}{ll}
\sigma_{e'}(e) & \textrm{ if } e \in E_{e'} \textrm{ for some } e' \in E_T;\\
\pi(e) & \textrm{ otherwise};\\
\end{array} \right.
\]

is a proper partial 8-edge-coloring of $G$ such that $\sigma^{-1}([8]) = \pi^{-1}([8]) \cup E_T$. We prove that $\sigma$ is proper by contradiction. Suppose there are incident edges $e_1, e_2 \in E$ such that $\sigma(e_1) = \sigma(e_2)$. If $e_1, e_2 \not\in E_e$ for every $e \in E_T$, then $\sigma(e_1) = \sigma(e_2)$ contradicts $\pi$ being a proper partial 8-edge-coloring. If $e_1 \in E_e$ for some $e \in E_T$ and $e_2 \not\in E_{e'}$ for every $e' \in E_T$ (or vice versa), then $\sigma(e_1) = \sigma(e_2)$ contradicts $\sigma_e$ being a proper partial 8-edge-coloring. When both $e_1, e_2 \in E_e$ for some $e \in E_T$, we use the same argument. Finally, if $e_1 \in E_e$ and $e_2 \in E_{e'}$ for some different $e,e' \in E_T$, then $\sigma(e_1) = \sigma(e_2)$ contradicts $\sigma_e$ and $\sigma_{e'}$ being coherent.

By another consequence of pairwise coherence, if $\sigma_e|_{E_e}$ is computed in time $O(|E_e|)$ for every $e \in E_T$, then the total time required to compute $\sigma$ is $O\left(\sum_{e \in E_T} |E_e|\right) = O(n)$. We can first compute $\sigma_e |_{E_e}$ for all $e \in E_T$ in time $O(n)$, and then compute $\sigma$ in time $O(n)$.

To summarize, to finish the proof, we will do the following:
\begin{itemize}
    \item for $e \in E_T$, give an algorithm finding $\sigma_e |_{E_e}$ that runs in time $O(|E_e|)$,
    \item for different $e, e' \in E_T$, show that $\sigma_e$ and $\sigma_{e'}$ are coherent.
\end{itemize}

We will consider multiple cases depending on $T$.

Assume $T = 0$. The number of edges at distance at most 1 from $e$ is $O(1)$ because $\Delta(G) \leq 8$. Hence, in constant time we can iterate over all possible ways to modify $\pi$ by coloring $e$ and recoloring some edges at distance at most 1 from $e$. We will find a proper 8-edge-coloring, $\sigma_e$, by the definition of type 0. Also, the edges in $E_T$ are 4-independent, so $E_e$ and $E_{e'}$ are 1-independent and thus $\sigma_e$ and $\sigma_{e'}$ are coherent.

Assume $T \neq 0$. We name the vertices giving $e$ type $T$ in the same way as in the definitions of the types. For example, if $T = 3_{ab}$ for different $a, b \in [8]$, then $e = xy$, where $a \in \bar{\pi}(x)$ and $a \in \pi(y)$, and $z \in N_G(x)$ satisfies $b \in \bar{\pi}(z)$ and $\pi(xz) = c$ for some $c \in \bar{\pi}(y)$. For any $T$ we can find the vertices giving $e$ type $T$ in constant time, since looking at the colors of edges at distance at most 1 is sufficient to identify them (note that we do not have to check conditions related to chains because we already know $e$ has type $T$). Let

\[
V_e = \left\{\begin{array}{ll}
\{x, y\} &  \textrm{ if } T = 1_{ab} \textrm{ for some } a, b \in [8];\\
\{x, y, z\} & \textrm{ if } T = 3_{ab} \textrm{ for some } a, b \in [8];\\
\{x, y, z, v\} & \textrm{ if } T \in \{4_{ab}, 5_{ab}\} \textrm{ for some } a, b \in [8];\\
\{x, y, z, v_1, v_2\} & \textrm{ if } T = 6_{abcd} \textrm{ for some } a, b, c, d \in [8];\\
\end{array} \right.
\]

Similarly, $\{x', y'\}$ if $T = 1_{ab}$ for some $a, b \in [8]$, etc.

For now, assume $T \in \{1_{ab}, 3_{ab}, 4_{ab}, 5_{ab}\}$ for some different $a, b \in [8]$. We will consider the case $T = 6_{abcd}$ for some $a, b, c, d \in [8]$ at the end. While presenting the algorithm finding $\sigma_e|_{E_e}$ for different $T$, we will additionally argue that for any $T$ the set $E_e$ is a union of:
\begin{itemize}
    \item $E_{V_e}$ --- a set of (not necessarily all) edges whose both endpoints are in $V_e$,
    \item $E_{P_e}$ --- a set of all the edges of a single $(a, b)$-chain $P_e$ that has $y$ or $z$ as an endpoint and does not contain an edge incident with any vertex of $V_{e'}$.
\end{itemize}

We now show that the condition above implies that $E_e$ and $E_{e'}$ are 0-independent, and thus $\sigma_e$ and $\sigma_{e'}$ are coherent. First, $E_{V_e}$ and $E_{V_{e'}}$ are 0-independent because the edges in $E_T$ are 4-independent and the edges in $E_{V_e}$ and $E_{V_{e'}}$ are at distance at most 1 from $e$ and $e'$, respectively. Second, $E_{P_e}$ and $E_{V_{e'}}$ are 0-independent by the definition of $E_{P_e}$. Similarly, $E_{V_e}$ and $E_{P_{e'}}$ are 0-independent. Finally, $E_{P_e}$ and $E_{P_{e'}}$ would not be 0-independent only if $P_e$ and $P_{e'}$ were the same chain, since both $P_e$ and $P_{e'}$ are $(a, b)$-chains. However, this case is impossible. Indeed, recall that $P_e$ and $P_{e'}$ respectively have an endpoint at distance at most 1 from $e$ and $e'$.
These endpoints are different by $4$-independence of the edges in $E_T$.
Hence if $P_e=P_{e'}$ then $e$ and $e'$ are $T$-chain dependent, a contradiction.

We show below how to find $\sigma_e|_{E_e}$ in time $O(|E_e|)$ for $T \in \{1_{ab}, 3_{ab}, 4_{ab}, 5_{ab}\}$ for different $a, b \in [8]$. We make sure that for each $T$, the chain $P_e$ has $y$ or $z$ as an endpoint. We also argue that $P_e$ does not contain an edge incident with any vertex of $V_{e'}$. We do not specify $E_{V_e}$ precisely, but it is easy to check that, apart from the edges in $E_{P_e}$, $\sigma_e$ differs from $\pi$ only for edges with both endpoints in $V_e$.

Assume $T = 1_{ab}$. To get $\sigma_e$ from $\pi$, we can kempe the $(a, b)$-chain $P_e$ starting at $y$ and set $\pi(xy) := a$. $P_e$ cannot contain an edge incident with $x'$ or $y'$ because $a \in \bar{\pi}(x')$, $b \in \bar{\pi}(y')$, and $P_e \neq P_{e'}$.

Assume $T = 3_{ab}$ for different $a, b \in [8]$. To get $\sigma_e$ from $\pi$, we can kempe the $(a, b)$-chain $P_e$ starting at $z$, set $\pi(xz) := a$, and set $\pi(xy) := c$. Note that in the definition of $3_{ab}$ we assume there is an $(a, b)$-chain with endpoints $x$ and $y$ to ensure that $P_e$ does not contain $x$ or $y$. $P_{e'}$ starts at $z'$ and $P_e \neq P_{e'}$, so $P_e$ cannot contain an edge incident with $z'$. There is an $(a, b)$-chain with endpoints $x'$ and $y'$, which is different from $P_e$ because the edges in $E_T$ are 4-independent. Hence, $P_e$ cannot contain an edge incident with $x'$ or $y'$.

Assume $T = 4_{ab}$. To get $\sigma_e$ from $\pi$, we can kempe the $(a, b)$-chain $P_e$ starting at $z$, swap$_{\pi}(xv, zv)$, and set $\pi(xy) := a$. $P_{e'}$ starts at $z'$ and $P_e \neq P_{e'}$, so $P_e$ cannot contain an edge incident with $z'$. Since $a, b \in \bar{\pi}(y')$, $P_e$ cannot contain an edge incident with $y'$. There is an $(a, b)$-cycle containing $x'$ and $v'$, which is different from $P_e$ because $P_e$ is not a cycle. Hence, $P_e$ cannot contain an edge incident with $x'$ or $v'$.

Assume $T = 5_{ab}$. To get $\sigma_e$ from $\pi$, we can kempe the $(a, b)$-chain $P_e$ starting at $z$ (in particular, we set $\pi(zv) := b$), swap$_{\pi}(xv, zv)$, and set $\pi(xy) := c$. $P_{e'}$ starts at $z'$, contains $v'$, and is different from $P_e$, so $P_e$ cannot contain an edge incident with $z'$ or $v'$. There is an $(a, b)$-chain with endpoints $x'$ and $y'$, which is different from $P_e$ because the edges in $E_T$ are 4-independent. Hence, $P_e$ cannot contain an edge incident with $x'$ or $y'$.

We are left with the case $T = 6_{abcd}$ for different $a, b, c, d \in [8]$. Here, $E_e = E_{V_e} \cup E_{P_e} \cup E_{R_e}$, where $E_{R_e}$ is the set of edges of an $(c, d)$-chain $R_e$ that has $z$ as an endpoint and does not contain an edge incident with any vertex of $V_{e'}$. In the definition of $\sigma_e$, we will kempe the chain $R_e$ just like $P_e$.

We prove that with the extended definition of $E_e$, $\sigma_e$ and $\sigma_{e'}$ are coherent. The proof that $E_{V_e} \cup E_{P_e}$ and $E_{V_{e'}} \cup E_{P_{e'}}$ are 0-independent remains the same as for $T \in \{1_{ab}, \ldots, 5_{ab}\}$. $E_{R_e}$ and $E_{V_{e'}}$ are 0-independent by the definition of $E_{R_e}$. Similarly, $E_{V_e}$ and $E_{R_{e'}}$ are 0-independent. Just as before, we can argue that $R_e$ and $R_{e'}$ must be different chains, and thus $E_{R_e}$ and $E_{R_{e'}}$ are 0-independent by the assumption that the edges in $E_T$ are $T$-independent. $E_{P_e}$ and $E_{R_{e'}}$ are not necessarily 1-independent, but $\sigma_e |_{E_{P_e}}$ and $\sigma_{e'} |_{E_{R_{e'}}}$ are coherent because $P_e$ is an $(a, b)$-chain, $R_{e'}$ is an $(c, d)$-chain, and we kempe chains $P_e$ and $R_{e'}$ in the definitions of $\sigma_e$ and $\sigma_{e'}$, respectively. Similarly, $\sigma_e |_{E_{R_e}}$ and $\sigma_{e'} |_{E_{P_{e'}}}$ are coherent.

We show below how to find $\sigma_e$ in time $O(|E_e|)$ for $T = 6_{abcd}$ for different $a, b, c, d \in [8]$. We ensure that $P_e$ and $R_e$ have $z$ as an endpoint and argue that these chains do not contain an edge incident with any vertex of $V_{e'}$.

To get $\sigma_e$ from $\pi$, we can kempe the $(a, b)$-chain $P_e$ starting at $z$, kempe the $(c, d)$-chain $R_e$ starting at $z$, swap$_{\pi}(xv_1, zv_1)$, and set $\pi(xy) := c$. $P_{e'}$ starts at $z'$, contains $v_1'$ and is different from $P_e$, so $P_e$ cannot contain an edge incident with $z'$ or $v_1'$. There is an $(a, b)$-chain with endpoints $x'$ and $y'$ that contains $v_2'$ and is different from $P_e$ because the edges in $E_T$ are 4-independent. Hence, $P_e$ cannot contain an edge incident with $x'$, $y'$, or $v_2'$. $R_{e'}$ starts at $z'$, contains $v_2'$ and is different from $R_e$, so $R_e$ cannot contain an edge incident with $z'$ or $v_2'$. Since $c, d \in \bar{\pi}(y')$, $R_e$ cannot contain an edge incident with $y'$. There is an $(c, d)$-cycle containing $x'$ and $v_1'$, which is different from $R_e$ because $R_e$ is not a cycle. Hence, $R_e$ cannot contain an edge incident with $x'$ or $v_1'$. This completes the proof.
\end{proof}

We can now state and prove the equivalent of \reftheorem{color-many-weak} for butterfly-like edges.

\begin{theorem}
    Let $G = (V, E)$ be a graph such that $\Delta(G) \leq 8$ and let $W \subseteq E$ be a set of butterfly-like edges that are pairwise 4-independent. Also, let $\pi$ be a partial 8-edge-coloring of $G$ such that $\pi^{-1}([8]) = E \setminus W$. We can find in time $O(|V|)$ a partial 8-edge-coloring $\sigma$ of $G$ such that $\sigma^{-1}([8]) = (E \setminus W) \cup W'$ for some $W' \subseteq W$ of size $\Omega(|W|)$.
\label{color-many-butterfly-like}
\end{theorem}

\begin{proof}
Let $n = |V|$. Note that $|E|=O(n)$ because $\Delta(G)=O(1)$. Since the edges in $W$ are 4-independent and $\pi^{-1}([8]) = E \setminus W$, all the edges at distance at most 1 from any edge of $W$ are colored. Hence, all the edges in $W$ have some type by \reflemma{types-butterfly}. Moreover, given $CDS(G, \pi)$, which we can compute in time $O(n)$ by \reflemma{compute-cds}, we can find the types of all the edges in $W$ in time $O(|W|) = O(n)$.

Let $T$ be a type that maximizes $|\{e \in W : e \textrm{ has type } T\}|$. The number of types is bounded by a constant, so there are $\Omega(|W|)$ edges of type $T$ in $W$. Denote the set of these edges as $E_T$.
Also, let
\[
T' = \left\{\begin{array}{ll}
1_{ab} &  \textrm{ if } T = 2_{ab} \textrm{ for some } a, b \in [8];\\
T & \textrm{ otherwise.}\\

\end{array} \right.
\]

If $T = 2_{ab}$ for some $a, b \in [8]$, we can modify $\pi$ in time $O(|E_T|)$ so that all the edges in $E_T$ gain type $1_{ab}$ by \reflemma{no-type-2}. Therefore, we can ensure in time $O(n)$ that all the edges in $E_{T'}$ have type $T'$ and $|E_{T'}| = \Omega(|W|)$. In case \reflemma{no-type-2} was applied, we also recompute $CDS(G, \pi)$.

The next step is to find $W' \subseteq E_{T'}$ such that the edges in $W'$ are $T'$-independent and $|W'| = \Omega(|E_{T'}|) = \Omega(|W|)$. We can take edges to $W'$ greedily --- we consider all the edges of $E_{T'}$ in any order and take an edge to $W'$ if and only if it is $T'$-independent with all the edges already in $W'$. We use $CDS(G, \pi)$ to check if an edge is $T'$-independent with all the edges in $W'$. By the definition of $T'$-independence, we only have to check that all chains with appropriate colors starting at distance at most 1 from the edge end at distance more than 1 from all the edges in $W'$. Since $\Delta(G) \leq 8$, the number of chains to check is at most $2 \cdot 2\Delta(G) \leq 32$ and for each chain we can perform the verification in constant time using $CDS(G, \pi)$ if we keep marking all vertices at distance at most 1 from edges taken to $W'$. We are guaranteed to take at least $\lceil |E_{T'}| / 33\rceil$ edges to $W'$ because any edge in $E_{T'}$ must be $T'$-independent with all but at most $32$ other edges in $E_{T'}$.

Finally, we can expand $\pi$ by all the edges in $W'$ in time $O(n)$ by \reflemma{single-type-color-many}.
\end{proof}

\subsection{The algorithm}

Now we collect the ingredients from previous sections to show our main result as follows.

\begin{theorem}
\label{th:main}
    There is an algorithm that runs in time $O(n\log n)$ and finds an 8-edge-coloring of an input $n$-vertex planar graph with maximum degree at most 8. The algorithm does not need a planar embedding.
\end{theorem}


\begin{proof}
Our algorithm is recursive.
Given $G$, we first remove all isolated vertices from it. If we end up with an empty graph, we return an empty coloring.

Then, we find the set of all weak edges $E_w$ and the set of all butterfly-like edges $E_b$. For any edge, we can decide whether it is weak or butterfly-like in constant time (recall that $\Delta(G)$ is bounded), so we can compute $E_w$ and $E_b$ in time $O(n)$.

Let $R$ be equal to $E_w$ if $|E_w| \geq |E_b|$ and $E_b$ otherwise. By \reftheorem{linear-bound-reducible-edges}, $|R| \geq n/2920$. Next, we find in time $O(n)$ any maximal set $I \subseteq R$ such that edges in $I$ are pairwise 4-independent. We take edges to $I$ greedily --- we consider all the edges in $R$ in any order and take an edge to $I$ if and only if it is 4-independent with all the edges already in $I$. To be able to check 4-independence with all the edges in $I$ in constant time, we keep marking all vertices at distance at most 4 from the edges taken to $I$. Since $\Delta(G) \leq 8$, any edge of $R$ can be at distance at most 4 from $O(1)$ other edges in $R$. Hence, $|I| = \Omega{(|R|)}$.

Next, we find an 8-edge-coloring $\pi$ of $G \setminus I$ recursively. Then, we repeatedly apply \reftheorem{color-many-weak} or \reftheorem{color-many-butterfly-like} to all uncolored edges of $G$ until we color all edges of $G$. In each step, the number of uncolored edges decreases by a constant factor, so there are $O(\log |I|)$ steps. Each of them takes time $O(n)$, which gives time $O(n \log |I|) = O(n \log n)$ in total. For each recursive call, the number of edges decreases by a constant factor $\alpha$, and the number of vertices remaining in the graph (after removing isolated vertices) is upper-bounded by twice the number of edges, so the total time of the algorithm is $O(n \log n)$, since
\[\sum_{i\ge 0} \alpha^i n \log (\alpha^i n) \le \left(\sum_{i\ge 0} \alpha^i\right) n \log n = O(n\log n).\]
\end{proof}

Let us also provide a generalization to graphs of small genus $g$. 

\begin{theorem}
	\label{thm:genus}
	There is an algorithm that runs in time $O((n+g)\log n + 2^{O(g)}g^{O(1)})$ for a given input genus $g$ graph $G$ with maximum degree at most 8 and finds an edge-coloring of $G$ using the optimal number of $\chi'(G)$ colors. The algorithm does not need an embedding.
\end{theorem}

\begin{proof}
We proceed as in the proof of Theorem~\ref{th:main}. Let us point out the differences. Since genus $g$ graphs have $O(n+g)$ edges, computing $E_w$ and $E_b$ takes time $O(n+g)$.
By Theorem~\ref{linear-bound-reducible-edges}, $|R|\ge n/2920-(g-1)/6$.
Note that when $g\ge 2$, for small graphs the right hand side of this inequality may not be positive and then it may happen that $R$ is actually empty and the algorithm does not make any progress in coloring more edges. This is not surprising, because for genus $g\ge 2$, it may happen that $G$ is not 8-edge-colorable.
Thus, when $n/2920-(g-1)/6\le n/5840$, i.e., when $n\le \tfrac{5840}6(g-1)$, we just find an optimal coloring using an exponential-time algorithm, say the vertex coloring algorithm of Bj\"{o}rklund et al.~\cite{DBLP:journals/siamcomp/BjorklundHK09} applied to the line graph of $G$. It runs in time $2^{|E(G)|}|E|^{O(1)} = 2^{O(g)}g^{O(1)}$.
Otherwise, we have $|R|>n/5840$ and we proceed as in the proof of Theorem~\ref{th:main}.
\end{proof}

\section{Further research}
\label{further-research}

One direction to improve the result from this work is to find an algorithm complexity than $o(n \log n)$, ideally $O(n)$. A natural idea is to extend the approach of Cole and Kowalik~\cite{cole-kowalik} to $\Delta = 8$. Their $O(n)$ algorithm subsequently expands the partial $\Delta$-edge-coloring by one edge in constant time. We failed to prove that for $\Delta=8$ there is always an edge reducible in constant time, even though we used computer-assisted case analysis.

Another open problem is finding a subquadratic algorithm for the case $\Delta=7$. We considered whether the approach from this work can also be used for $\Delta=7$, but we realized it is not easy. The first challenge is to show that there are $\Omega(n)$ reducible (perhaps in a different sense) edges. The proof of Sanders and Zhao~\cite{SandersZhaoDelta7} only shows that there is at least one reducible edge when $\Delta = 7$. The second challenge is to find a way to efficiently expand a partial 7-edge-coloring by $\Omega(n)$ reducible edges. Expanding a partial 7-edge-coloring by a single reducible edge in the proof of Sanders and Zhao may look very similar to what we did in the proof of \reflemma{single-type-color-many} --- recoloring only a constant number of edges in a close neighborhood of the edge and kemping a constant number of chains. However, in the proof of Sanders and Zhao, kemping chains that share a color is possible. For example, one may have to kempe $(a, b)$-chain $P_1$ before kemping $(b, c)$-chain $P_2$. Kemping $P_1$ will change $P_2$ if these chains share an edge colored $b$. In particular, $P_2$ may be merged with another $(b, c)$-chain $P_3$. This significantly complicates the proof of \reflemma{single-type-color-many} for the case $\Delta = 7$ because we cannot rely anymore on coherence of edges from $E_T$. For example, some edges of $P_3$ could later be recolored while expanding the 7-edge-coloring by other edges from $E_T$.

\bibliographystyle{plainurl}
\bibliography{planar-edgecol8}

\end{document}